\documentclass[preliminary,copyright,creativecommons,sharealike]{eptcs}
\providecommand{\event}{FICS 2015}

\usepackage{mysty}

\title{$^*$-Continuous Kleene $\omega$-Algebras for Energy
  Problems\thanks{%
    The work of the first author was supported by the National
    Foundation of Hungary for Scientific Research, Grant no.\
    K~108448.  The work of the second and third authors was supported
    by ANR MALTHY, grant no.~ANR-13-INSE-0003 from the French National
    Research Foundation, and by the EU FP7 SENSATION project, grant
    no.\ 318490 (FP7-ICT-2011-8).}}

\author{%
  Zolt{\'a}n {\'E}sik%
  \institute{University of Szeged, Hungary}%
  \and Uli Fahrenberg \qquad Axel Legay%
  \institute{Inria Rennes, France}%
}

\newcommand*\authorrunning{Z.~{\'E}sik, U.~Fahrenberg and A.~Legay}
\newcommand*\titlerunning{$^*$-Continuous Kleene $\omega$-algebras for
  energy problems}

\begin{document}

\maketitle

\begin{abstract}
  Energy problems are important in the formal analysis of embedded or
  autonomous systems.  Using recent results on $^*$-continuous Kleene
  $\omega$-algebras, we show here that energy problems can be solved
  by algebraic manipulations on the transition matrix of energy
  automata.  To this end, we prove general results about certain
  classes of finitely additive functions on complete lattices which
  should be of a more general interest.
\end{abstract}

\section{Introduction}

Energy problems are concerned with the question whether a given system
admits infinite schedules during which (1) certain tasks can be
repeatedly accomplished and (2) the system never runs out of energy
(or other specified resources).  These are important in areas such as
embedded systems or autonomous systems and, starting
with~\cite{DBLP:conf/formats/BouyerFLMS08}, have attracted some
attention in recent years, for example
in~\cite{DBLP:conf/ictac/FahrenbergJLS11, DBLP:conf/lata/Quaas11,
  DBLP:conf/hybrid/BouyerFLM10, DBLP:journals/pe/BouyerLM14,
  DBLP:journals/iandc/VelnerC0HRR15, DBLP:journals/tcs/ChatterjeeD12,
  DBLP:conf/hybrid/BrenguierCR14, DBLP:conf/birthday/JuhlLR13,
  DBLP:conf/csl/DegorreDGRT10}.

With the purpose of generalizing some of the above approaches, we have
in~\cite{DBLP:conf/atva/EsikFLQ13, conf/wata/FahrenbergLQ12}
introduced \emph{energy automata}.  These are finite automata whose
transitions are labeled with \emph{energy functions} which specify how
energy values change from one system state to another.  Using the
theory of semiring-weighted automata~\cite{book/DrosteKV09}, we have
shown in~\cite{DBLP:conf/atva/EsikFLQ13} that energy problems in such
automata can be solved in a simple static way which only involves
manipulations of energy functions.

In order to put the work of~\cite{DBLP:conf/atva/EsikFLQ13} on a more
solid theoretical footing and with an eye to future generalizations,
we have recently introduced a new algebraic structure of
\emph{$^*$-continuous Kleene
  $\omega$-algebras}~\cite{DBLP:conf/dlt/EsikFL15} (see
also~\cite{DBLP:journals/corr/EsikFL15} for the long version).  We
show here that energy functions form such a $^*$-continuous Kleene
$\omega$-algebra.  Using the fact, proven
in~\cite{DBLP:conf/dlt/EsikFL15}, that for automata with transition
weights in $^*$-continuous Kleene $\omega$-algebras, reachability and
B{\"u}chi acceptance can be computed by algebraic manipulations on the
transition matrix of the automaton, the results
from~\cite{DBLP:conf/atva/EsikFLQ13} follow.

\section{Energy Automata}
\label{se:energy}

The transition labels on the energy automata which we consider in the
paper, will be functions which model transformations of energy levels
between system states.  Such transformations have the (natural)
properties that below a certain energy level, the transition might be
disabled (not enough energy is available to perform the transition), and
an increase in input energy always yields at least the same increase in
output energy.  Thus the following definition:

\begin{definition}
  \label{de:energy}
  An \emph{energy function} is a partial function $f: \Realnn\parto
  \Realnn$ which is defined on a closed interval $[ l_f,
  \infty\mathclose[$ or on an open interval $\mathopen] l_f,
  \infty\mathclose[$, for some lower bound $l_f\ge 0$, and such that for
  all $x\le y$ for which $f$ is defined,
  \begin{equation}
    \label{eq:deriv1}
    y f\ge x f+ y- x\,.
  \end{equation}
  The class of all energy functions is denoted by $\mathcal F$.
\end{definition}

Note that we write function composition and application in
diagrammatical order, \emph{from left to right}, in this paper.  Hence
we write $f;g$, or simply $f g$, for the composition $g\circ f$ and
$x; f$ or $x f$ for function application $f( x)$.  This is because we
will be concerned with \emph{algebras} of functions, in which function
composition is multiplication, and where it is customary to write
multiplication in diagrammatical order.

Thus energy functions are strictly increasing, and in points where
they are differentiable, the derivative is at least $1$.  The inverse
functions to energy functions exist, but are generally not energy
functions.  Energy functions can be \emph{composed}, where it is
understood that for a composition $f g$, the interval of definition is
$\{ x\in \Realnn\mid x f\text{ and } x f g\text{ defined}\}$. 

\begin{lemma}
  \label{le:inf-strict}
  Let $f\in \mathcal F$ and $x\in \Realnn$. If $x f< x$, then there is
  $N\in \Nat$ for which $x f^N$ is not defined.  If $x f> x$, then for
  all $P\in \Real$ there is $N\in \Nat$ for which $x f^N\ge P$.
\end{lemma}

\begin{proof}
  In the first case, we have $x- x f= M> 0$.  Using~\eqref{eq:deriv1},
  we see that $x f^{ n+ 1}\le x f^n- M$ for all $n\in \Nat$ for which
  $x f^{ n+ 1}$ is defined.  Hence $( x f^n)_{ n\in \Nat}$ decreases
  without bound, so that there must be $N\in \Nat$ such that $x f^N$
  is undefined.
  
  In the second case, we have $x f- x= M> 0$.  Again
  using~\eqref{eq:deriv1}, we see that $x f^{ n+ 1}> x f^n+ M$ for all
  $n\in \Nat$.  Hence $( x f^n)_{ n\in \Nat}$ increases without bound,
  so that for any $P\in \Real$ there must be $N\in \Nat$ for which
  $x f^N\ge P$.
\end{proof}

Note that property~\eqref{eq:deriv1} is not only sufficient for
Lemma~\ref{le:inf-strict}, but in a sense also necessary: if
$0< \alpha< 1$ and $f: \Realnn\to \Realnn$ is the function
$x f= 1+ \alpha x$, then
$x f^n= \sum_{ i= 0}^{ n- 1} \alpha^i+ \alpha^n x$ for all
$n\in \Nat$, hence $\lim_{ n\to \infty} x f^n= \frac1{ 1- \alpha}$, so
Lemma~\ref{le:inf-strict} does not hold for $f$.  On the other hand,
$y f= x f+ \alpha( y- x)$ for all $x\le y$, so~\eqref{eq:deriv1}
``almost'' holds.

\begin{definition}
  An \emph{energy automaton} $( S, s_0, T, F)$ consists of a finite
  set $S$ of states, with initial state $s_0\in S$, a finite set
  $T\subseteq S\times \mathcal F\times S$ of transitions labeled with
  energy functions, and a subset $F\subseteq S$ of acceptance states.
\end{definition}

\begin{figure}[tbp]
  \centering
  \begin{tikzpicture}[->,>=stealth',shorten >=1pt,auto,node
    distance=2.0cm,initial text=,scale=1]
    \tikzstyle{every node}=[font=\small]
    \tikzstyle{every state}=[fill=white,shape=circle,inner
    sep=.5mm,minimum size=4mm,outer sep=.1mm]
    \node[state, initial] (1) at (0,0) {};
    \node[state, accepting] (2) at (4,0) {};
    \node[state] (3) at (8,0) {};
    \path (1) edge[out=10,in=170] node[above] {$x\mapsto x+ 2; x\ge 2$}
    (2);
    \path (1) edge[out=-10,in=-170] node[below] {$x\mapsto x+ 3; x> 1$}
    (2);
    \path (2) edge[out=120,in=60,loop] node[above] {$x\mapsto 2x- 2; x\ge
      1$} (2);
    \path (2) edge[out=10,in=170] node[above] {$x\mapsto x- 1; x> 1$}
    (3);
    \path (3) edge[out=190,in=-10] node[below] {$x\mapsto x+ 1;
      x\ge 0$}
    (2);
  \end{tikzpicture}
  \caption{%
    \label{fi:eauto}
    A simple energy automaton.}
\end{figure}
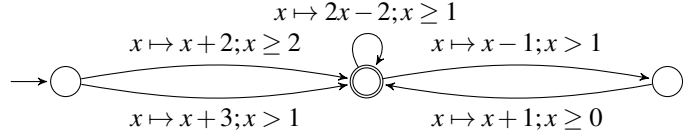

We show an example of a simple energy automaton in Fig.~\ref{fi:eauto}.
Here we use inequalities to give the definition intervals of energy
functions.

A finite \emph{path} in an energy automaton is a finite sequence of
transitions
$\pi= (s_0,f_1,s_1), (s_1,f_2,s_2),\dotsc,$ $(s_{n-1},f_n,s_n)$.  We use
$f_\pi$ to denote the combined energy function $f_1 f_2\dotsm f_n$ of
such a finite path.  We will also use infinite paths, but note that
these generally do not allow for combined energy functions.

A \emph{global state} of an energy automaton is a pair $q=( s, x)$ with
$s\in S$ and $x\in \Realnn$.  A transition between global states is of
the form $((s,x), f, (s',x'))$ such that $(s,f,s')\in T$ and $x'=f(x)$.
A (finite or infinite) \emph{run} of $( S, T)$ is a path in the graph of
global states and transitions.

We are ready to state the decision problems with which our main concern
will lie.  As the input to a decision problem must be in some way
finitely representable, we will state them for subclasses $\mathcal
F'\subseteq \mathcal F$ of \emph{computable} energy functions; an
$\mathcal F'$-automaton is an energy automaton $( S, T)$ with
$T\subseteq S\times \mathcal F'\times S$.

\begin{problem}[Reachability]
  \label{pb:reach}
  Given a subset $\F'\subseteq \F$ of computable functions, an
  $\F'$-automaton $A=( S, s_0, T, F)$ and a computable initial energy
  $x_0\in \Realnn$: does there exist a finite run of $A$ from
  $(s_0, x_0)$ which ends in a state in $F$?
\end{problem}

\begin{problem}[B{\"u}chi acceptance]
  \label{pb:buchi}
  Given a subset $\F'\subseteq \F$ of computable functions, an
  $\F'$-automaton $A=( S, s_0, T, F)$ and a computable initial energy
  $x_0\in \Realnn$: does there exist an infinite run of $A$ from
  $(s_0, x_0)$ which visits $F$ infinitely often?
\end{problem}

As customary, a run such as in the statements above is said to be
accepting.


\section{Algebraic Preliminaries}

We now turn our attention to the algebraic setting of $^*$-continuous
Kleene algebras and related structures, before revisiting energy
automata in Section~\ref{se:energy2}.  In this section we review some
results on $^*$-continuous Kleene algebras and $^*$-continuous Kleene
$\omega$-algebras.

\subsection{$^*$-Continuous Kleene $\omega$-Algebras}
\label{se:starcontkleom}

A \emph{semiring}~\cite{book/BerstelR10, book/Golan99} $S =
(S,+,\cdot,0,1)$ consists of a commutative monoid $(S,+,0)$ and a
monoid $(S,\cdot,1)$ such that the distributive laws
\begin{align*}
  x(y+ z) & = xy + xz\\
  (y+z)x &= yx + zx
\end{align*}
and the zero laws
\begin{equation*}
  0\cdot x = 0 = x \cdot 0
\end{equation*}
hold for all $x,y,z \in S$. It follows that the product operation
distributes over all finite sums.

An \emph{idempotent semiring} is a semiring $S$ whose sum operation is
idempotent, so that $x + x = x$ for all $x \in S$. Each idempotent
semiring $S$ is partially ordered by the relation $x\le y$ iff $x + y
= y$, and then sum and product preserve the partial order and $0$ is
the least element. Moreover, for all $x,y \in S$, $x + y$ is the least
upper bound of the set $\{x,y\}$. Accordingly, in an idempotent
semiring $S$, we will usually denote the sum operation by $\vee$ and
$0$ by $\bot$.

A \emph{Kleene algebra}~\cite{DBLP:journals/iandc/Kozen94} is an
idempotent semiring $S = (S,\vee, \cdot, \bot,1)$ equipped with a star
operation $^*: S \to S$ such that for all $x,y\in S$, $yx^*$ is the
least solution of the fixed point equation $z = zx \vee y$ and $x^*y$
is the least solution of the fixed point equation $z = xz \vee y$ with
respect to the natural order.

A \emph{$^*$-continuous Kleene
  algebra}~\cite{DBLP:journals/iandc/Kozen94} is a Kleene algebra $S=(
S, \vee, \cdot, ^*, \bot, 1)$ in which the infinite suprema $\bigvee\{
x^n\mid n\ge 0\}$ exist for all $x\in S$, $x^*= \bigvee\{ x^n\mid n\ge
0\}$ for every $x\in S$, and product preserves such suprema:
\begin{equation*}
  y \big( \bigvee_{ n\ge 0} x^n \big)= \bigvee_{ n\ge 0} y x^n\quad
  \text{and}\quad \big( \bigvee_{ n\ge 0} x^n \big) y= \bigvee_{ n\ge
    0} x^n y
\end{equation*}
for all $x, y\in S$.

A \emph{continuous Kleene algebra} is a Kleene algebra
$S=( S, \vee, \cdot, ^*, \bot, 1)$ in which \emph{all} suprema
$\bigvee X$, $X\subseteq S$, exist and are preserved by products,
\ie~$y( \bigvee X)= \bigvee yX$ and $( \bigvee X)y= \bigvee Xy$ for
all $X\subseteq S$, $y\in S$.  $^*$-continuous Kleene algebras are
hence a generalization of continuous Kleene algebras.  There are
interesting Kleene algebras which are $^*$-continuous but not
continuous, for example the Kleene algebra of all regular languages
over some alphabet.

A \emph{semiring-semimodule pair}~\cite{journals/sgf/EsikK07,
  book/BloomE93} $( S, V)$ consists of a semiring $S =
(S,+,\cdot,0,1)$ and a commutative monoid $V = (V,+,0)$ which is
equipped with a left $S$-action $S \times V \to V$, $(s,v) \mapsto
sv$, satisfying
\begin{alignat*}{2}
  (s+s')v   &= sv + s'v \qquad\qquad&
  s(v + v') &= sv + sv' \\
  (ss') v   &= s(s'v) &
  0s &= 0 \\
  s0 &= 0 &
  1v &= v
\end{alignat*}
for all $s,s'\in S$ and $v \in V$.  In that case, we also call $V$ a
\emph{(left) $S$-semimodule}.  If $S$ is idempotent, then also $V$ is
idempotent, so that we then write $V=(V,\vee,\bot)$.

A \emph{generalized $^*$-continuous Kleene
  algebra}~\cite{DBLP:conf/dlt/EsikFL15} is a semiring-semimodule pair $(
S, V)$ where $S=(S,\vee,\cdot, ^*,\bot,1)$ is a $^*$-continuous Kleene
algebra such that
\begin{equation*}
  x y^* v= \bigvee_{ n\ge 0} x y^n v
\end{equation*}
for all $x, y\in S$ and $v\in V$.

A \emph{$^*$-continuous Kleene
  $\omega$-algebra}~\cite{DBLP:conf/dlt/EsikFL15} consists of a generalized
$^*$-continuous Kleene algebra $( S, V)$ together with an infinite
product operation $S^\omega\to V$ which maps every infinite sequence
$x_0, x_1,\dotsc$ in $S$ to an element $\prod_{ n\ge 0} x_n$ of $V$.
The infinite product is subject to the following conditions:
\begin{itemize}
\item[\Ax1] For all $x_0, x_1,\dotsc\in S$, $\displaystyle \prod_{
    n\ge 0} x_n= x_0 \prod_{ n\ge 0} x_{ n+ 1}$.
\item[\Ax2] Let $x_0, x_1,\dotsc\in S$ and $0= n_0\le n_1\le\dotsm$ a
  sequence which increases without a bound. Let $y_k= x_{ n_k}\dotsm
  x_{ n_{ k+ 1}- 1}$ for all $k\ge 0$.  Then $\displaystyle \prod_{
    n\ge 0} x_n= \prod_{ k\ge 0} y_k$.
\item[\Ax3] For all $x_0, x_1,\dotsc, y, z\in S$, $\displaystyle
  \prod_{ n\ge 0}( x_n( y\vee z))= \adjustlimits \bigvee_{ x_0',
    x_1',\dotsc\in\{ y, z\}\;} \prod_{ n\ge 0} x_n x_n'$.
\item[\Ax4] For all $x, y_0, y_1,\dotsc\in S$, $\displaystyle \prod_{
    n\ge 0} x^* y_n= \adjustlimits \bigvee_{ k_0, k_1,\dotsc\ge 0\;}
  \prod_{ n\ge 0} x^{ k_n} y_n$.
\end{itemize}

A \emph{continuous Kleene
  $\omega$-algebra}~\cite{journals/sgf/EsikK07} is a
semiring-semimodule pair $( S, V)$ in which $S$ is a continuous Kleene
algebra, $V$ is a complete lattice, and the $S$-action on $V$
preserves all suprema in either argument, together with an infinite
product as above which satisfies conditions~\Ax1 and ~\Ax2 above and
preserves all suprema:
$\prod_{ n\ge 0}( \bigvee X_n)= \bigvee \{ \prod_{ n\ge 0} x_n\mid
x_n\in X_n, n\ge 0\}$
for all $X_0, X_1,\dotsc\subseteq S$ (this property implies~\Ax3
and~\Ax4 above).  $^*$-continuous Kleene $\omega$-algebras are hence a
generalization of continuous Kleene $\omega$-algebras.  We have
in~\cite{DBLP:conf/dlt/EsikFL15} given an example, based on regular
languages of finite and infinite words, of a $^*$-continuous Kleene
$\omega$-algebra which is not a continuous Kleene $\omega$-algebra.
In Section~\ref{se:energy2} we will show that energy functions give
raise to another such example.

\subsection{Matrix Semiring-Semimodule Pairs}

For any semiring $S$ and $n\ge 1$, we can form the matrix semiring
$S^{ n\times n}$ whose elements are $n\times n$-matrices of elements
of $S$ and whose sum and product are given as the usual matrix sum and
product.  It is known~\cite{DBLP:conf/mfcs/Kozen90} that when $S$ is a
$^*$-continuous Kleene algebra, then $S^{ n\times n}$ is also a
$^*$-continuous Kleene algebra, with the $^*$-operation defined by
\begin{equation*}
  M^*_{ i, j}= \adjustlimits \bigvee_{ m\ge 0\;} \bigvee_{ 1\le
    k_1,\dotsc, k_m\le n} M_{ i, k_1} M_{ k_1, k_2}\dotsm M_{ k_m, j}
\end{equation*}
for all $M\in S^{ n\times n}$ and $1\le i, j\le n$.  The above
infinite supremum exists, as it is taken over a regular set,
see~\cite[Thm.~9]{DBLP:journals/jucs/EsikK02}
and~\cite[Lemma~4]{DBLP:conf/dlt/EsikFL15}.  Also, if $n\ge 2$ and
$M= \left( \begin{smallmatrix} a & b \\ c & d \end{smallmatrix}
\right)$,
where $a$ and $d$ are square matrices of dimension less than $n$, then
\begin{equation}
  \label{eq:mstar}
  M^*= 
  \begin{pmatrix}
    ( a\vee b d^* c)^* & ( a\vee b d^* c)^* b d^* \\
    ( d\vee c a^* b)^* c a^* & ( d\vee c a^* b)^*
  \end{pmatrix}\,.
\end{equation}

For any semiring-semimodule pair $( S, V)$ and $n\ge 1$, we can form
the matrix semiring-semimodule pair $( S^{ n\times n}, V^n)$ whose
elements are $n\times n$-matrices of elements of $S$ and
$n$-dimensional (column) vectors of elements of $V$, with the action
of $S^{ n\times n}$ on $V^n$ given by the usual matrix-vector product.

When $( S, V)$ is a $^*$-continuous Kleene $\omega$-algebra, then
$( S^{ n\times n}, V^n)$ is a generalized $^*$-continuous Kleene
algebra~\cite{DBLP:conf/dlt/EsikFL15}.
By~\cite[Lemma~17]{DBLP:conf/dlt/EsikFL15}, there is an $\omega$-operation
on $S^{ n\times n}$ defined by
\begin{equation*}
  M^\omega_i= \bigvee_{1\le k_1,k_2,\dotsc\le n} M_{ i, k_1} M_{ k_1, k_2}\dotsm
\end{equation*}
for all $M\in S^{ n\times n}$ and $1\le i\le n$.  Also, if $n\ge 2$
and $M= \left( \begin{smallmatrix} a & b \\ c & d \end{smallmatrix}
\right)$, where $a$ and $d$ are square matrices of dimension less than
$n$, then
\begin{equation*}
  M^\omega= 
  \begin{pmatrix}
    ( a\vee b d^* c)^\omega\vee( a\vee b d^* c)^* b d^\omega \\
    ( d\vee c a^* b)^\omega\vee( d\vee c a^* b)^* c a^\omega
  \end{pmatrix}\,.
\end{equation*}

\subsection{Weighted automata}
\label{se:weightedaut}

Let $( S, V)$ be a $^*$-continuous Kleene $\omega$-algebra and
$A\subseteq S$ a subset.  We write $\langle A\rangle$ for the set of
all finite suprema $a_1\vee\dotsm\vee a_m$ with $a_i\in A$ for each
$i= 1,\dotsc, m$.

A \emph{weighted automaton}~\cite{inbook/EsikK09} over $A$ of
dimension $n\ge 1$ is a tuple $( \alpha, M, k)$, where
$\alpha\in\{ \bot, 1\}^n$ is the initial vector,
$M\in \langle A \rangle^{ n\times n}$ is the transition matrix, and
$k$ is an integer $0\le k\le n$.  Combinatorially, this may be
represented as a transition system whose set of states is
$\{ 1,\dotsc, n\}$.  For any pair of states $i, j$, the transitions
from $i$ to $j$ are determined by the entry $M_{ i, j}$ of the
transition matrix: if $M_{ i, j}= a_1\vee\dotsm\vee a_m$, then there
are $m$ transitions from $i$ to $j$, respectively labeled
$a_1,\dotsc, a_n$.  The states $i$ with $\alpha_i= 1$ are
\emph{initial}, and the states $\{ 1,\dotsc, k\}$ are
\emph{accepting}.

The \emph{finite behavior} of a weighted automaton $A=( \alpha, M, k)$
is defined to be
\begin{equation*}
  | A|= \alpha M^* \kappa\,,
\end{equation*}
where $\kappa\in\{ \bot, 1\}^n$ is the vector given by $\kappa_i= 1$
for $i\le k$ and $\kappa_i= \bot$ for $i> k$.  (Note that $\alpha$ has
to be used as a \emph{row} vector for this multiplication to make
sense.)  It is clear by~\eqref{eq:mstar} that $| A|$ is the supremum
of the products of the transition labels along all paths in $A$ from
any initial to any accepting state.

The \emph{B{\"u}chi behavior} of a weighted automaton
$A=( \alpha, M, k)$ is defined to be
\begin{equation*}
  \| A\|= \alpha
  \begin{pmatrix}
    ( a+ b d^* c)^\omega \\
    d^* c( a+ b d^* c)^\omega
  \end{pmatrix},
\end{equation*}
where $a\in \langle A\rangle^{ k\times k}$,
$b\in \langle A\rangle^{ k\times( n- k)}$,
$c\in \langle A\rangle^{( n- k)\times n}$ and
$d\in \langle A\rangle^{( n- k)\times( n- k)}$ are such that
$M= \left( \begin{smallmatrix} a & b \\ c &
    d \end{smallmatrix}\right)$.
By~\cite[Thm.~20]{DBLP:conf/dlt/EsikFL15}, $\| A\|$ is the supremum of the
products of the transition labels along all infinite paths in $A$ from
any initial state which infinitely often visit an accepting state.

\section{Generalized $^*$-continuous Kleene Algebras of Functions}
\label{se:finadd}

In the following two sections our aim is to establish properties which
ensure that semiring-semimodule pairs of functions form
$^*$-continuous Kleene $\omega$-algebras.  We will use these
properties in Section~\ref{se:energy2} to show that energy functions
form a $^*$-continuous Kleene $\omega$-algebra.

Let $L$ and $L'$ be complete lattices with bottom and top elements
$\bot$ and $\top$.  Then a function $f: L\to L'$ is said to be
\emph{finitely additive} if $\bot f = \bot$ and
$( x\vee y) f= x f\vee y f$ for all $x, y\in L$.  (Recall that we
write function application and composition in the diagrammatic order,
from left to right.)  When $f: L\to L'$ is finitely additive, then
$( \bigvee X) f= \bigvee X f$ for all finite sets $X \subseteq L$.

Consider the collection $\FinAdd_{ L, L'}$ of all finitely additive
functions $f: L\to L'$, ordered pointwise.  Since the (pointwise)
supremum of any set of finitely additive functions is finitely
additive, $\FinAdd_{ L, L'}$ is also a complete lattice, in which the
supremum of any set of functions can be constructed pointwise.  The
least and greatest elements are the constant functions with value
$\bot$ and $\top$, respectively.  By an abuse of notation, we will
denote these functions by $\bot$ and $\top$ as well.

\begin{definition}
  A function $f\in \FinAdd_{ L, L'}$ is said to be
  \emph{$\top$-continuous} if $f= \bot$ or for all $X\subseteq L$ with
  $\bigvee X= \top$, also $\bigvee X f= \top$.
\end{definition}

Note that if $f\ne \bot$ is $\top$-continuous, then $\top f= \top$.
The functions $\id$ and $\bot$ are $\top$-continuous.  Also, the
(pointwise) supremum of any set of $\top$-continuous functions is
again $\top$-continuous.

We will first be concerned with functions in $\FinAdd_{ L, L}$, which
we just denote $\FinAdd_L$.  Since the composition of finitely
additive functions is finitely additive and the identity function
$\id$ over $L$ is finitely additive, and since composition of finitely
additive functions distributes over finite suprema, $\FinAdd_L$,
equipped with the operation $\vee$ (binary supremum), $;$
(composition), and the constant function $\bot$ and the identity
function $\id$ as $1$, is an idempotent semiring.  It follows that
when $f$ is finitely additive, then so is $f^*= \bigvee_{ n\ge 0}
f^n$.  Moreover, $f\le f^*$ and $f^*\le g^*$ whenever $f\le g$.  Below
we will usually write just $fg$ for the composition $f;g$.

\begin{lemma}
  \label{le:rightcont}
  Let $S$ be any subsemiring of\/ $\FinAdd_L$ closed under the
  $^*$-operation. Then $S$ is a $^*$-continuous Kleene algebra iff for
  all $g, h\in S$, $g^* h= \bigvee_{ n\ge 0} g^n h$.
\end{lemma}

\begin{proof}
  Suppose that the above condition holds.  We need to show that $f(
  \bigvee_{ n\ge 0} g^n) h= \bigvee_{ n\ge 0} f g^n h$ for all $f, g,
  h\in S$.  But $f( \bigvee_{ n\ge 0} g^n) h= f( \bigvee_{ n\ge 0} g^n
  h)$ by assumption, and we conclude that $f( \bigvee_{ n\ge 0} g^n
  h)= \bigvee_{ n\ge 0} f g^n h$ since the supremum is pointwise.
\end{proof}

Compositions of $\top$-continuous functions in $\FinAdd_L$ are again
$\top$-continuous, so that the collection of all $\top$-continuous
functions in $\FinAdd_L$ is itself an idempotent semiring.

\begin{definition}
  A function $f\in \FinAdd_L$ is said to be \emph{locally $^*$-closed}
  if for each $x\in L$, either $x f^*= \top$ or there exists $N\ge 0$
  such that $x f^*= x\vee\dotsm\vee x f^N$.
\end{definition}

The functions $\id$ and $\bot$ are locally $^*$-closed.  As the next
example demonstrates, compositions of locally $^*$-closed (and
$\top$-continuous) functions are not necessarily locally $^*$-closed.

\begin{example}
  Let $L$ be the following complete lattice (the linear sum of three
  infinite chains):
  \begin{equation*}
    \bot< x_0< x_1<\dotsm< y_0< y_1<\dotsm< z_0< z_1<\dotsm< \top
  \end{equation*}
  Since $L$ is a chain, a function $L\to L$ is finitely additive iff
  it is monotone and preserves $\bot$.

  Let $f, g: L\to L$ be the following functions.  First, $\bot f= \bot
  g= \bot$ and $\top f =\top g= \top$.  Moreover, $x_i f= y_i$, $y_i
  f= z_i g= \top$ and $x_i g =\bot$, $y_i g= x_{i+1}$, and $z_i g=
  \top$ for all $i$.  Then $f, g$ are monotone, $u f^* =u\vee u f\vee
  u f^2$ and $u g^*= u\vee u g$ for all $u\in L$.  Also, $f$ and $g$
  are $\top$-continuous, since if $\bigvee X= \top$ then either $\top
  \in X$ or $X\cap\{ z_0, z_1,\dotsc\}$ is infinite, but then $\bigvee
  X f= \bigvee X g= \top$.  However, $f g$ is not locally $^*$-closed,
  since $x_0( f g)^*= x_0\vee x_0( f g)\vee x_0( f g)^2\dotsm= x_0\vee
  x_1\vee\dotsm= y_0$.
\end{example}

\begin{lemma}
  \label{le:starstar}
  Let $f\in \FinAdd_L$ be locally $^*$-closed.  Then also $f^*$ is
  locally $^*$-closed.  If $f$ is additionally $\top$-continuous, then
  so is $f^*$.
\end{lemma}

\begin{proof}
  We prove that $x f^{ **}= x\vee x f^*= x f^*$ for all $x\in L$.
  Indeed, this is clear when $x f^*= \top$, since $f^*\le f^{ **}$.
  Otherwise $x f^*= \bigvee_{ k\le n} x f^k$ for some $n\ge 0$.

  By finite additivity, it follows that
  $x f^* f^*= \bigvee_{ k\le n} x f^k f^*$.  But for each $k$,
  $x f^k f^*= x f^k\vee x f^{k+1}\vee\dotsm\le x f^*$, thus
  $x f^*= x f^* f^*$ and $x f^*=x f^{ **}$.  It follows that $f^*$ is
  locally $^*$-closed.

  Suppose now that $f$ is additionally $\top$-continuous.  We need to
  show that $f^*$ is also $\top$-continuous.  To this end, let
  $X\subseteq L$ with $\bigvee X= \top$.  Since $x\le x f^*$ for all
  $x\in X$, it holds that $\bigvee X f^*\ge \bigvee X= \top$.  Thus
  $\bigvee X f^* = \top$.
\end{proof}

\begin{proposition}
  \label{pr:lclotop}
  Let $S$ be any subsemiring of\/ $\FinAdd_L$ closed under the
  $^*$-operation.  If each $f\in S$ is locally $^*$-closed and
  $\top$-continuous, then $S$ is a $^*$-continuous Kleene algebra.
\end{proposition}

\begin{proof}
  Suppose that $g, h\in S$. By Lemma~\ref{le:rightcont}, it suffices
  to show that $g^* h= \bigvee_{ n\ge 0} g^n h$.  Since this is clear
  when $h= \bot$, assume that $h\ne \bot$.  As $g^n h\le g^* h$ for
  all $n\ge 0$, it holds that $\bigvee_{ n\ge 0} g^n h\le g^* h$. To
  prove the opposite inequality, suppose that $x\in L$. If
  $x g^*= \top$, then $\bigvee_{ n\ge 0} x g^n= \top$, so
  $\bigvee_{ n\ge 0} x g^n h= \top$ by $\top$-continuity.  Thus,
  $x g^* h= \top= \bigvee_{ n\ge 0} x g^n h$.

  Suppose that $x g^*\ne \top$.  Then there is $m\ge 0$ with
  \begin{equation*}
    x g^* h= ( x\vee\dotsm\vee x g^m) h= x h\vee\dotsm\vee x g^m h\le
    \bigvee_{ n\ge 0} x g^n h= x( \bigvee_{ n\ge 0} g^n h)\,.
  \end{equation*}
\end{proof}

Now define a left action of $\FinAdd_L$ on $\FinAdd_{ L, L'}$ by $f v=
f; v$, for all $f\in \FinAdd_L$ and $v\in \FinAdd_{ L, L'}$.  It is a
routine matter to check that $\FinAdd_{ L, L'}$, equipped with the
above action, the binary supremum operation $\vee$ and the constant
$\bot$ is an (idempotent) left $\FinAdd_L$-semimodule, that is, $(
\FinAdd_L, \FinAdd_{ L, L'})$ is a semiring-semimodule pair.

\begin{lemma}
  Let $S\subseteq \FinAdd_L$ be a $^*$-continuous Kleene algebra and
  $V\subseteq \FinAdd_{ L, L'}$ an $S$-semimodule.  Then $( S, V)$ is
  a generalized $^*$-continuous Kleene algebra iff for all $f\in S$
  and $v\in V$, $f^* v= \bigvee_{ n\ge 0} f^n v$.
\end{lemma}

\begin{proof}
  Similar to the proof of Lemma~\ref{le:rightcont}
\end{proof}

\begin{proposition}
  \label{pr:genstarkle}
  Let $S\subseteq \FinAdd_L$ be a $^*$-continuous Kleene algebra and
  $V\subseteq \FinAdd_{ L, L'}$ an $S$-semimodule.  If each $f\in S$
  is locally $^*$-closed and $\top$-continuous and each $v\in V$ is
  $\top$-continuous, then $( S, V)$ is a generalized $^*$-continuous
  Kleene algebra.
\end{proposition}

\begin{proof}
  Similar to the proof of Proposition~\ref{pr:lclotop}.
\end{proof}

\section{$^*$-continuous Kleene $\omega$-Algebras of Functions}
\label{se:finaddomega}

In this section, let $L$ be an arbitrary complete lattice and $L'=
\two$, the $2$-element lattice $\{ \bot, \top\}$.  We define an
infinite product $\FinAdd_L^\omega\to \FinAdd_{ L, \two}$.  Let $f_0,
f_1,\dotsc\in \FinAdd_L$ be an infinite sequence and define $v=
\prod_{n\ge 0} f_n: L\to \two$ by
\begin{equation*}
  x v=
  \begin{cases}
    \bot &\text{if there is $n\ge 0$ such that $x f_0\dotsm f_n=
      \bot$}, \\
    \top &\text{otherwise}
  \end{cases}
\end{equation*}
for all $x\in L$.  We will write $\prod_{ n\ge k} f_n$, for $k\ge 0$,
as a shorthand for $\prod_{ n\ge 0} f_{ n+ k}$.

It is easy to see that $\prod_{ n\ge 0} f_n$ is finitely additive.
Indeed, $\bot \prod_{ n\ge 0} f_n= \bot$ clearly holds, and for all
$x\le y\in L$, $x \prod_{ n\ge 0} f_n\le y \prod_{ n\ge 0} f_n$.
Thus, to prove that
$( x\vee y) \prod_{ n\ge 0} f_n= x \prod_{ n\ge 0} f_n\vee y \prod_{
  n\ge 0} f_n$
for all $x, y\in L$, it suffices to show that if
$x \prod_{ n\ge 0} f_n= y \prod_{ n\ge 0} f_n= \bot$, then
$( x\vee y) \prod_{ n\ge 0} f_n= \bot$.  But if
$x \prod_{ n\ge 0} f_n= y \prod_{ n\ge 0} f_n= \bot$, then there exist
$m, k\ge 0$ such that $x f_0\dotsm f_m= y f_0\dotsm f_k= \bot$. Let
$n = \max\{ m, k\}$. We have
$( x\vee y) f_0\dotsm f_n= x f_0\dotsm f_n\vee y f_0\dotsm f_n= \bot$,
and thus $( x\vee y) \prod_{ n\ge 0} f_n = \bot$.

It is clear that this infinite product satisfies conditions~\Ax1
and~\Ax2 in the definition of $^*$-continuous Kleene
$\omega$-algebra.  Below we show that also~\Ax3 and~\Ax4 hold.

\begin{lemma}
  \label{le:ax3}
  For all $f_0, f_1,\dotsc, g_0, g_1,\dotsc\in \FinAdd_L$,
  \begin{equation*}
    \prod_{ n\ge 0}( f_n\vee g_n)= \bigvee_{ h_n\in\{ f_n, g_n\}}
    \prod_{ n\ge 0} h_n\,.
  \end{equation*}
\end{lemma}

\begin{proof}
  Since infinite product is monotone, the term on the right-hand side
  of the equation is less than or equal to the term on the left-hand
  side.  To prove that equality holds, let $x\in L$ and suppose that
  $x \prod_{ n\ge 0}( f_n\vee g_n)= \top$.  It suffices to show that
  there is a choice of the functions $h_n\in\{ f_n, g_n\}$ such that
  $x \prod_{ n\ge 0} h_n= \top$.

  Consider the infinite ordered binary tree where each node at level
  $n\ge 0$ is the source of an edge labeled $f_n$ and an edge labeled
  $g_n$, ordered as indicated.  We can assign to each node $u$ the
  composition $h_u$ of the functions that occur as the labels of the
  edges along the unique path from the root to that node.

  Let us mark a node $u$ if $x h_u\ne \bot$.  As
  $x \prod_{ n\ge 0}( f_n\vee g_n)= \top $, each level contains a
  marked node.  Moreover, whenever a node is marked and has a
  predecessor, its predecessor is also marked.  By K{\"o}nig's
  lemma~\cite{Konig27} there is an infinite path going through marked
  nodes.  This infinite path gives rise to the sequence
  $h_0, h_1,\dotsc$ with $x \prod_{ n\ge 0} h_n= \top$.
\end{proof}

\begin{lemma}
  \label{le:ax4-help}
  Let $f\in \FinAdd_L$ and $v\in \FinAdd_{ L, \two}$ such that $f$ is
  locally $^*$-closed and $v$ is $\top$-continuous.  If $x f^* v=
  \top$, then there exists $k\ge 0$ such that $x f^k v= \top$.
\end{lemma}

\begin{proof}
  If $x f^*= \bigvee_{ n= 0}^N x f^n$ for some $N\ge 0$, then
  $x f^* v= \bigvee_{ n= 0}^N x f^n v= \top$ implies the claim of the
  lemma.  If $x f^*= \top$, then $\top$-continuity of $v$ implies that
  $\bigvee_{ n\ge 0} x f^n v= \top$, which again implies the claim.
\end{proof}

\begin{lemma}
  \label{le:ax4}
  Let $f, g_0, g_1,\dotsc\in \FinAdd_L$ be locally $^*$-closed and
  $\top$-continuous such that for each $m\ge 0$,
  $g_m \prod_{ n\ge m+ 1} f^* g_n\in \FinAdd_{ L, \two}$ is
  $\top$-continuous.  Then
  \begin{equation*}
    \prod_{ n\ge 0} f^* g_n= \adjustlimits \bigvee_{ k_0,
      k_1,\dotsc\ge 0\;} \prod_{ n\ge 0} f^{k_n} g_n\,.
  \end{equation*}
\end{lemma}

\begin{proof}
  As infinite product is monotone, the term on the right-hand side of
  the equation is less than or equal to the term on the left-hand
  side.  To prove that equality holds, let $x\in L$ and suppose that
  $x \prod_{ n\ge 0} f^* g_n= \top$. We want to show that there exist
  integers $k_0, k_1,\dotsc\ge 0$ such that
  $x \prod_{ n\ge 0} f^{k_n} g_n= \top$.

  Let $x_0= x$.  By Lemma~\ref{le:ax4-help},
  $x \prod_{ n\ge 0} f^* g_n= x_0 f^* g_0 \prod_{ n\ge 1} f^* g_n=
  \top$
  implies that there is $k_0\ge 0$ for which
  $x_0 f^{ k_0} g_0 \prod_{ n\ge 1} f^* g_n= \top$.  We finish the
  proof by induction.  Assume we have $k_0,\dotsc, k_m\ge 0$ such that
  $x f^{ k_0} g_0\dotsm f^{ k_m} g_m \prod_{ n\ge m+ 1} f^* g_n= \top$
  and let $x_{ m+ 1}= x f^{ k_0} g_0\dotsm f^{ k_m} g_m$.  Then
  $x_{ m+ 1} f^* g_{ m+ 1} \prod_{ n\ge m+ 2} f^* g_n= \top$ implies, using
  Lemma~\ref{le:ax4-help}, that there exists $k_{ m+ 1}\ge 0$ for
  which $x_{ m+ 1} f^{ k_{ m+ 1}} g_{ m+ 1} \prod_{ n\ge m+ 2} f^*
  g_n= \top$.
\end{proof}

\begin{proposition}
  \label{pr:starcontkleom}
  Let $S\subseteq \FinAdd_L$ and $V\subseteq \FinAdd_{ L, \two}$ such
  that $( S, V)$ is a generalized $^*$-continuous Kleene algebra of
  locally $^*$-closed and $\top$-continuous functions $L\to L$ and
  $\top$-continuous functions $L\to \two$.  If\/
  $\prod_{ n\ge 0} f_n\in V$ for all sequences $f_0, f_1,\dotsc$ of
  functions in $S$, then $( S, V)$ is a $^*$-continuous Kleene
  $\omega$-algebra.
\end{proposition}

\begin{proof}
  This is clear from Lemmas~\ref{le:ax3} and~\ref{le:ax4}.
\end{proof}

We finish the section by a lemma which exhibits a condition on the
lattice $L$ which ensures that infinite products of locally
$^*$-closed and $\top$-continuous functions are again
$\top$-continuous.

\begin{lemma}
  \label{le:infprodenergy}
  Assume that $L$ has the property that whenever $\bigvee X= \top$ for
  some $X\subseteq L$, then for all $x< \top$ in $L$ there is $y\in X$
  with $x\le y$.  If $f_0, f_1,\dotsc\in \FinAdd_L$ is a sequence of
  locally $^*$-closed and $\top$-continuous functions, then
  $\prod_{ n\ge 0} f_n\in \FinAdd_{ L, \two}$ is $\top$-continuous.
\end{lemma}

\begin{proof}
  Let $v= \prod_{ n\ge 0} f_n$.  We already know that $v$ is finitely
  additive.  We need to show that if $v\ne \bot$, then $v$ is
  $\top$-continuous.  But if $v\ne \bot$, then there is some $x< \top$
  with $x v= \top$, \ie~such that $x f_0\dotsm f_n> \bot$ for all $n$.
  By assumption, there is some $y\in X$ with $x\le y$.  It follows
  that $y f_0\dotsm f_n\ge x f_0\dotsm f_n> \bot$ for all $n$ and thus
  $\bigvee X v= \top$.
\end{proof}

\section{Energy Automata Revisited}
\label{se:energy2}

We finish this paper by showing how the setting developed in the last
sections can be applied to solve the energy problems of
Section~\ref{se:energy}.  Let $L=[ 0, \top]_\bot$ be the complete
lattice of nonnegative real numbers together with $\top= \infty$ and
an extra bottom element $\bot$, and extend the usual order and
operations on real numbers to $L$ by declaring that $\bot< x< \top$,
$\bot- x= \bot$ and $\top+ x= \top$ for all $x\in \Realnn$.  Note that
$L$ satisfies the precondition of Lemma~\ref{le:infprodenergy}.

We extend the definition of energy function:

\begin{definition}
  An \emph{extended energy function} is a mapping $f: L\to L$ for
  which $\bot f= \bot$, $\top f= \bot$ if $x f= \bot$ for all
  $x< \top$ and $\top f= \top$ otherwise, and $y f\ge x f+ y- x$
  whenever $\bot< x< y< \top$.  The set of such functions is denoted
  $\E$.
\end{definition}

Every energy function $f: \Realnn\parto \Realnn$ as of
Definition~\ref{de:energy} gives rise to an extended energy function
$\tilde f: L\to L$ given by $\bot \tilde f= \bot$, $x \tilde f= \bot$
if $x f$ is undefined, $x \tilde f= x f$ otherwise for $x\in \Realnn$,
and $\top \tilde f= \top$.  This defines an embedding
$\F\hookrightarrow \E$.

The definition entails that for all $f\in \E$ and all $x< y\in L$,
$x f= \top$ implies $y f= \top$ and $y f= \bot$ implies $x f= \bot$.
Note that $\E$ is closed under (pointwise) binary supremum $\vee$ and
composition and contains the functions $\bot$ and $\id$.

\begin{lemma}
  Extended energy functions are finitely additive and
  $\top$-continuous, hence $\E\subseteq \FinAdd_L$ is a semiring.
\end{lemma}

\begin{proof}
  Finite additivity follows from monotonicity.  For $\top$-continuity,
  let $X\subseteq L$ such that $\bigvee X= \top$ and $f\in \E$,
  $f\ne \bot$.  We have $X\ne \{ \bot\}$, so let
  $x_0\in X\setminus\{ \bot\}$ and, for all $n\ge 0$, $x_n= x_0+ n$.
  Let $y_n= x_n f$.  If $y_n= \bot$ for all $n\ge 0$, then also
  $n f= \bot$ for all $n\ge 0$ (as $x_n\ge n$), hence $f= \bot$.  We
  must thus have an index $N$ for which $y_N> \bot$.  But then
  $y_{ N+ k}\ge y_N+ k$ for all $k\ge 0$, hence $\bigvee X f= \top$.
\end{proof}

\begin{lemma}
  For $f\in \E$, $f^*$ is given by $x f^*= x$ if $x f\le x$ and
  $x f^*= \top$ if $x f> x$.  Hence $f$ is locally $^*$-closed and
  $f^*\in \E$.
\end{lemma}

\begin{proof}
  We have $\bot f^*= \bot$ and $\top f^*= \top$.  Let
  $x\ne \bot, \top$.  If $x f\le x$, then $x f^n\le x$ for all
  $n\ge 0$, so that $x\le \bigvee_{ n\ge 0} x f^n\le x$, whence
  $x f^*= x$.  If $x f> x$, then let $a= x f- x> 0$.  We have
  $x f\ge x+ a$, hence $x f^n\ge x+ n a$ for all $n\ge 0$, so that
  $x f^*= \bigvee_{ n\ge 0} x f^n= \top$.
\end{proof}

Not all locally $^*$-closed functions $f: L\to L$ are energy
functions: the function $f$ defined by $x f= 1$ for $x< 1$ and
$x f= x$ for $x\ge 1$ is locally $^*$-closed, but $f\notin \E$.

\begin{corollary}
  $\E$ is a $^*$-continuous Kleene algebra.
\end{corollary}

\begin{proof}
  This is clear by Proposition~\ref{pr:lclotop}.
\end{proof}

\begin{remark}
  It is \emph{not} true that $\E$ is a \emph{continuous} Kleene
  algebra: Let $f_n, g\in \E$ be defined by
  $x f_n= x+ 1- \frac1{ n+ 1}$ for $x\ge 0$, $n\ge 0$ and $x g= x$ for
  $x\ge 1$, $x g= \bot$ for $x< 1$.  Then
  $0( \bigvee_{ n\ge 0} f_n) g=( \bigvee_{ n\ge 0} 0 f_n) g= 1 g= 1$,
  whereas
  $0 \bigvee_{ n\ge 0}( f_n g)= \bigvee_{ n\ge 0}( 0 f_n g)= \bigvee_{
    n\ge 0}(( 1- \frac1{ n+ 1}) g)= \bot$.
\end{remark}

Let $\V$ denote the $\E$-semimodule of all $\top$-continuous functions
$L\to \two$.  For $f_0, f_1,\dotsc\in \E$, define the infinite product
$f= \prod_{ n\ge 0} f_n:L\to \two$ by $x f= \bot$ if there is an index
$n$ for which $x f_0\dotsm f_n= \bot$ and $x f= \top$ otherwise, like
in Section~\ref{se:finaddomega}.  By Lemma~\ref{le:infprodenergy},
$\prod_{ n\ge 0} f_n$ is $\top$-continuous,
\ie~$\prod_{ n\ge 0} f_n\in \V$.

By Proposition~\ref{pr:genstarkle}, $( \E, \V)$ is a generalized
$^*$-continuous Kleene algebra.

\begin{corollary}
  $( \E, \V)$ is a $^*$-continuous Kleene $\omega$-algebra.
\end{corollary}

\begin{proof}
  This is clear by Proposition~\ref{pr:starcontkleom}.
\end{proof}

\begin{remark}
  As $\E$ is not a continuous Kleene algebra, it also holds that
  $( \E, \V)$ is not a continuous Kleene $\omega$-algebra; in fact it
  is clear that there is no $\E$-semimodule $\V'$ for which
  $( \E, \V)$ would be a continuous Kleene $\omega$-algebra.  The
  initial motivation for the work in~\cite{DBLP:conf/dlt/EsikFL15} and the
  present paper was to generalize the theory of continuous Kleene
  $\omega$-algebras so that it would be applicable to energy
  functions.
\end{remark}

Noting that energy automata are weighted automata over $\E$ in the
sense of Section~\ref{se:weightedaut}, we can now solve the
reachability and B{\"u}chi problem for energy automata:

\begin{theorem}
  Let $A=( \alpha, M, k)$ be an energy automaton and $x_0\in \Realnn$.
  There exists a finite run of $A$ from an initial state to an
  accepting state with initial energy $x_0$ iff $x_0| A|> \bot$.
\end{theorem}

\begin{theorem}
  Let $A=( \alpha, M, k)$ be an energy automaton and $x_0\in
  \Realnn$.  There exists an infinite run of $A$ from an initial state
  which infinitely often visits an accepting state iff $x_0\| A\|=
  \top$.
\end{theorem}

\begin{corollary}
  Problems~\ref{pb:reach} and~\ref{pb:buchi} are decidable.
\end{corollary}

In~\cite{DBLP:conf/atva/EsikFLQ13}, the complexity of the decision
procedure has been established for important subclasses of energy
functions.

\section{Conclusion and Further Work}

We have shown that energy functions form a $^*$-continuous Kleene
$\omega$-algebra~\cite{DBLP:conf/dlt/EsikFL15}, hence that
$^*$-continuous Kleene $\omega$-algebras provide a proper algebraic
setting for energy problems.  On our way, we have proven more general
results about properties of finitely additive functions on complete
lattices which should be of a more general interest.

There are interesting generalizations of our setting of energy
automata which, we believe, can be attacked using techniques similar
to ours.  One such generalization are energy problems for \emph{real
  time} or \emph{hybrid} models, as for example treated
in~\cite{DBLP:conf/hybrid/BouyerFLM10, DBLP:conf/formats/BouyerFLMS08,
  DBLP:journals/pe/BouyerLM14, DBLP:conf/lata/Quaas11}.  Another
generalization is to higher dimensions, like
in~\cite{DBLP:conf/ictac/FahrenbergJLS11, DBLP:conf/birthday/JuhlLR13,
  DBLP:journals/iandc/VelnerC0HRR15} and other papers.

\bibliographystyle{plain}
\bibliography{mybib}

\end{document}